\documentclass{llncs}%{article}

\usepackage{mathpartir,amsmath,amssymb,bm,color}
\usepackage{kro}
\usepackage{enumerate}
\usepackage{fullpage}

\newif\iflong \longfalse

\title{On characterising strong bisimilarity\\
  in a fragment of CCS with replication\\
  -- note --}
\date{}
\author{Daniel Hirschkoff\inst{1} and Damien Pous\inst{2}}
\institute{ENS Lyon, Universit\'e de Lyon, CNRS, INRIA \and SARDES,
  LIG, Grenoble, CNRS, INRIA} 

\begin{document}

\maketitle

\begin{abstract}
  We provide a characterisation of strong bisimilarity in a fragment
  of CCS that contains only prefix, parallel composition,
  synchronisation and a limited form of replication. The
  characterisation is not an axiomatisation, but is instead presented
  as a rewriting system.

  We discuss how our method allows us to derive a new congruence
  result in the $\pi$-calculus: congruence holds in the sub-calculus
  that does not include restriction nor sum, and features a limited
  form of replication. We have not formalised the latter result in all
  details.
\end{abstract}

\section{Introduction}

We study algebraic properties of strong bisimilarity in a sub-calculus
of CCS. Like in previous work~\cite{hirschkoff:pous:lmcs:08}, of which
the present study is a continuation, an important aspect of the
setting we analyse is the absence of the sum construct, and, more
generally, of any operator that would allow us to decompose parallel
composition.

We present a rewriting system that allows us to characterise 
% what is actually a failed attempt at obtaining a simple
% axiomatisation of
strong bisimilarity ($\sim$) in a very basic calculus that only
features prefixes, parallel composition, and replicated prefixes, with
the additional constraint that these can occur only at top-level. The
restriction and choice (or sum) operators are not included.  Handling
replication is the novel aspect w.r.t.\
\cite{hirschkoff:pous:lmcs:08}, and raises several difficulties when
trying to analyse the algebraic properties of $\sim$.

Let us focus on the properties of replication w.r.t.\ strong
bisimilarity. In our setting, the most important bisimilarity law for
replication is written
\begin{mathpar}
  \!a.P\,|\,a.P ~=~ \!a.P
  \enspace,
\end{mathpar}
\noindent and expresses that a replicated process acts as an
unbounded number of copies of that process in parallel.

It appears that we can generalise the above equality, by allowing a
replicated process to erase one of its copies (we are reading the
equality from left to right here) not only at top-level, but
arbitrarily deep in a term. In other words, if $C$ is a context (a
process with a hole), the law
\begin{mathpar}
  \!a.P\,|\,C[a.P] ~=~ \!a.P\,|\,C[\nil]
\end{mathpar}
\noindent should hold for strong bisimilarity (the previous equality is
obtained by taking $C=[\,]$).

This equality, together with the law $\!a.P|\!a.P = \!a.P$, are the
basic ingredients we need in order to characterise strong bisimilarity
between replicated terms. However, these equations are not enough, as
the following example shows: process $P_1 = \!a.(b|a.c)|\!a.(c|a.b)$
is bisimilar to $P_2 = \!a.b|\!a.c$. It seems reasonable to consider
$P_2$ as the normal form of $P_1$. Intuitively, $P_1$ can be obtained
from $P_2$ by inserting a copy of $a.b$ ``inside'' $\!a.c$, and,
symmetrically, a copy of $a.c$ inside $\!a.b$.  A related difficulty
appears with equalities like $\!a.(b|a.b) = \!a.b$, where the copy is
inserted in the replicated component itself.

Describing this phenomenon of ``mutual replication'' in all its
generality would lead to complicated equational schemata, and we have
not been able to come up with a simple, readable, presentation of
strong bisimilarity based on equational laws. Instead, we introduce a
rewriting relation on processes that allows us to compute normal forms
w.r.t.\ strong bisimilarity (in particular, we are able to rewrite
$P_1$ into $P_2$). This has the advantage of exposing the basic laws
that are at work when normalising a process. We show that our
characterisation of strong bisimilarity still holds when we enrich the
calculus with synchronisation. In turn, the method we describe can be
applied to derive a new congruence result on a subset of the
$\pi$-calculus (we must say we have not checked all details of this
result yet).

\medskip\noindent\textsl{Outline.}  We describe the subset of CCS we
work with in Sect.~\ref{sec:set}; in Sect.~\ref{sec:tec}, we introduce
a notion of normal forms and prove useful some technical results. The
rewriting system is defined in Sect.~\ref{sec:rwt}, where we show that
it allows us to reach normal forms. Section~\ref{sec:comm} is devoted
to the extension of our results to a calculus with synchronisations,
closer to the standard CCS. In Sect.~\ref{sec:piccl}, we give
concluding remarks, discussing in particular how these results lead to
a new congruence property in the $\pi$-calculus.

\section{The Setting}
\label{sec:set}

We work in the subset of CCS defined by the following grammar, where
we rely on a countable set of \emph{actions} $\alpha,\beta,\dots$:
\begin{align*}
  F ~&::=~ \nil \OR \alpha.F \OR F|F &
  P,Q ~&::=~ F \OR \!\alpha.F \OR P|P \tag{processes}\\
  D ~&::=~ [] \OR \alpha.D \OR D|F &
  C ~&::=~ D \OR \!\alpha.D \OR C|P \tag{contexts}
%  P ~~::=~~ \nil\OR P_1|P_2\OR \alpha.P\OR \!\alpha.P
\end{align*}
Our calculus features no communication, no restriction, no sum, and
allows replication only on prefixes, at top-level. We use $P,Q$ to
range over processes. A \emph{finite} process $(F)$ is a process which
does not contain an occurrence of the replication operator.
For $F=\alpha_1.F_1\,|\dots|\,\alpha_k.F_k$, we shall sometime write $F$ as
$\prod_{i\in[1..k]} \alpha_i.F_i$, and denote by $\!F$ or
$\prod_{i\in[1..k]} \!\alpha_i.F_i$ the process
$\!\alpha_1.F_1\,|\dots|\!\alpha_k.F_k$.
%, and $\prod_{i\in\emptyset} \alpha_i.P_i$ will
% stand for \nil. 
Note that $\!F$ will always denote a process having replicated
components only.

We use $C$ to range over single-hole \emph{contexts} mapping finite
processes to processes. Accordingly, we use $D$ to range over
(single-hole) \emph{finite contexts}, mapping finite processes to
finite processes. Note that the hole cannot occur directly under a
replication in $C$.

The labelled transition system associated to this process calculus is
standard (Fig.~\ref{fig:lts} -- note that there is no synchronisation
rule, this will be addressed in Sect.~\ref{sec:comm}), and
yields a notion of \emph{strong bisimilarity}, written $\sim$, which
is a congruence.
\begin{figure}[t]
  \centering
  \begin{mathpar}
    \inferrule{ }{\alpha.F \xr \alpha F} \and
    \inferrule{F_1 \xr\alpha F'_1}{F_1\,|\,F_2 \xr \alpha F'_1\,|\,F_2} \and
    \inferrule{F_2 \xr\alpha F'_2}{F_1\,|\,F_2 \xr \alpha F_1\,|\,F'_2} \\
    \inferrule{ }{\!\alpha.F \xr \alpha \!\alpha.F\,|\, F} \and
    \inferrule{P_1 \xr\alpha P'_1}{P_1\,|\,P_2 \xr \alpha P'_1\,|\,P_2} \and
    \inferrule{P_2 \xr\alpha P'_2}{P_1\,|\,P_2 \xr \alpha P_1\,|\,P'_2} 
  \end{mathpar}
  \caption{Labelled Transition System for our Subset of CCS}
  \label{fig:lts}
\end{figure}

%
% Structural congruence, written $\equiv$, is the smallest congruence
% satisfying the laws of an abelian monoid for $|$, where \nil{} is the
% neutral element. A process is said \emph{finite} if it is structurally
% congruent to a process that does not contain any occurrence of
% replication. Any finite process is structurally congruent to either
% \nil{} or a parallel composition of prefixed processes. We shall
% sometimes write $\alpha_1.P_1\,|\dots|\,\alpha_k.P_k$ as $\prod_{1\leq i\leq k}
% \alpha_i.P_i$.

We shall rely on the following characterisation of strong bismilarity
for finite processes, which is established
in~\cite{hirschkoff:pous:lmcs:08}:

\begin{defi}[Distribution law]
  \label{def:distrlaw}
  Let \eqD{} be the smallest congruence generated by the laws of an
  abelian monoid for parallel composition (the neutral element being
  \nil), and the following equation schema, called \emph{distribution
    law}, where there are as many occurrences of $F$ on both sides of
  the equation.
  \begin{mathpar}
    \alpha.(F|\alpha.F|\dots|\alpha.F) ~=~ \alpha.F|\alpha.F|\dots|\alpha.F
    \enspace,
  \end{mathpar}
\end{defi}
It is easy to show that this congruence is decidable, and we have
\begin{thm}
  \label{thm:distrlaw}
  \eqD{} coincides with strong bisimilarity $(\sim)$ on finite processes.
\end{thm}

\section{Preliminary Technical Results}
\label{sec:tec}

We present some technical results about strong bisimilarity. Most of
these help us isolating the replicated part from the finite part in
processes being compared. Indeed, when characterising strong
bisimilarity, we shall prove that $P\sim Q$ implies that the
replicated parts of $P$ and $Q$ are bisimilar, and we also need
somehow to reason about the finite parts of $P$ and $Q$.

% The following proposition shows that copies of a replicated process
% can be inserted anywhere, in any term containing this replicated
% process modulo bisimilarity. 
The following property is necessary to derive correction of the
rewrite system we define below.
\begin{prop}\label{prop:simplelaw}
  If $C[\nil]\sim\!\alpha.F|P$, then $C[\nil]\sim C[\alpha.F]$.
\end{prop}
\begin{proof}
  We show that $\R = \set{(C[\nil], C[\alpha.F])~/~\forall C\text{
      s.t.\ } C[\nil]\sim\!\alpha.F|P\textrm{ for some }P}$ is a
  strong bisimulation up to transitivity and parallel composition
  (cf.~\cite{SW01,phd:pous}).

  There are three cases to consider in the bisimulation game:
  \begin{itemize}
  \item the hole occurs at top-level in the context ($C=[]|Q$) and the
    right-hand side process does the following transition:
    $C[\alpha.F]\xr\alpha F|Q$. By hypothesis, $Q\sim\!\alpha.F|P$ so
    that we find $Q'$ such that $Q\xr\alpha Q'$ and $Q'\sim
    \!\alpha.F|F|P$. By injecting the latter equality, we obtain
    $Q'\sim Q|F$ so that $Q'$ closes the diagram.
  \item the hole occurs under a replicated prefix in the context
    ($C=\!\beta.D|Q$) and this prefix is fired: we have
    $C[\nil]\xr\beta P_l=C[\nil]|D[\nil]$ and $C[\alpha.F]\xr\beta
    P_r=C[\alpha.F]|D[\alpha.F]$. This is where we need the up-to
    technique: these processes are not related by $\R$ (recall that we
    work with single-hole contexts). However, we can deduce $P_1\R
    P_c=C[\nil]|D[\alpha.F]$, by considering the context
    $C'=C[\nil]|D[]$, and checking that
    $C'[\nil]\sim\!\alpha.F|P|D[\nil]$. We finally check that $P_c$
    and $P_r$ are related by the closure of $\R$ under parallel
    contexts (by removing the $D[\alpha.F]$ component).
  \item in the last case, either the hole occurs under a
    non-replicated prefix in the contexts ($C=\beta.D|Q$), or the
    contexts triggers a transition that does not involve or duplicate
    the hole; this case is treated by a simple reasoning -- just play
    the bisimulation game.  \qed
  \end{itemize}
\end{proof}

As a consequence, we obtain the validity of the following laws:
\begin{mathpar}
  \!\alpha.F~|~C[\alpha.F] ~\sim~ \!\alpha.F~|~C[0] \quad(A)\and
  \!\alpha.D[\alpha.D[\nil]] ~\sim~ \!\alpha.D[\nil] \quad(A')
\end{mathpar}

\begin{lem}\label{lem:cancel:finite:part}
  If $\!F\sim \alpha.F'|Q$, then $\!F\sim\!F|\alpha.F'$.
\end{lem}
\begin{proof}
  Purely algebraically: replicate everything and add $\!\alpha.F'$ in
  parallel, this yields $\!F|\!\alpha.F' \sim \!\alpha.F'|\!Q|\!\alpha.F'$,
  from which we deduce $\!F|\!\alpha.F' \sim \!\alpha.F'|\!Q\sim
  \!\alpha.F'|\!Q|\alpha.F'\sim \!F|\alpha.F'$.
  (Note that, when writing $\!Q$, we actually refer to the process
  obtained by adding replication at top-level on the finite components
  of $Q$; we easily show that this operation preserves bisimilarity.)
  \qed
\end{proof}

\begin{lem}\label{lem:pouslemma}
  If $F=\prod_i \alpha_i.F_i$ and $\!F\sim\!\alpha.F'|Q$, then there
  exists $j$ s.t.\ $\!F\sim \!\alpha.F'~|~\prod_{i\neq
    j}\!\alpha_i.F_i$ and $\alpha_j=\alpha$.
\end{lem}
\begin{proof}
  By firing $\alpha.F'$ on the right-hand side, we find $j$ such that
  $\alpha_j=\alpha$ and $\!F|F_j\sim \!\alpha.F'|F'|Q$, from which we
  deduce $\!F|F_j\sim \!F|F'$. Then we show that the singleton
  relation $\{(\!F,\!\alpha.F'|\prod_{i\neq j}\!\alpha_i.F_i)\}$ is a
  bisimulation up to bisimilarity and parallel contexts.
  \begin{itemize}
  \item when a transition on $\alpha_i$ is triggered, with $i\neq j$,
    we reason up to parallel composition in order to remove the $F_i$
    component on both sides;
  \item when a transition on $\alpha_j$ (or $\alpha$) is triggered, we
    have to relate processes $\!F|F_j$ and
    $\!\alpha.F'|F'|\prod_{i\neq j}\!\alpha_i.F_i$ ; we reason up to
    bisimilarity in order to rewrite $\!F|F_j$ into $\!F|F'$ and then
    up to parallel context in order to remove the $F'$ component. \qed
  \end{itemize}
\end{proof}

\medskip

Now we define our notion of normal forms (\emph{seeds}).
\begin{defi}[Size, seed] %, notations]
  The size of $P$, noted \size{P}, is the number of prefixes in $P$.

  A \emph{seed} of $P$, noted \seed{P} is a process of minimal size
  such that $P\sim\seed{P}$.
  % The \emph{seed} of $\alpha.P$, noted \seed{\alpha.P}, is a
  % \textbf{finite} process of minimal size such that
  % $\!\alpha.P\visbis \!\alpha.\seed{P}$.
\end{defi}

The seed of a process is defined modulo bisimilarity. We establish in
this section that all seeds of a process are actually equated by
\eqD{} (Prop.~\ref{prop:compare:seeds}).  Note that, because $\sim$ is
a congruence in our calculus, if $P_1|P_2$ is a seed, then $P_1$ is a
seed. Indeed, if $\size{P'_1}<\size{P_1}$ and $P'_1\sim P_1$, then
$P_1|P_2\sim P'_1|P_2$, which contradicts the fact that $P_1|P_2$ is a
seed.

\medskip
\noindent \textbf{Notations.} 

We shall use $S, S'$ to range over seeds having only replicated
components.
%
% From now on, write $P=\prod_j \alpha_j.P_j$, and
% suppose $\!P\sim S = \prod_i \!\alpha_i.S_i$.
%
We write $P\reduct k Q$ whenever there exist $\alpha_1,..,\alpha_k$
and $P_0,..,P_k$ such that $P=P_0\xr{\alpha_1}P_1\dots
\xr{\alpha_k}P_k\eqD Q$.
Note that $P\reduct k \alpha.F$ for some $k$ if and only if
$P\eqD D[\alpha.F]$ for some finite context $D$.
For $S = \prod_i \!\alpha_i.S_i$, we write $S\dis F$ to denote the
fact that $\neg(\exists i,k,\, F\reduct k \alpha_i.S_i)$, i.e., that
$F$ does not contain a sub-term of the form $\alpha_i.S_i$. On the
contrary, we write $S\purg F$ when there exists $k>0$ such that
$S\reduct k S|F$, that is, when $F$ is a parallel composition of
sub-terms of the $S_i$s. In the sequel, we shall use $R$ to range over
finite processes satisfying the latter property.

We can remark that if $S\dis F$ (resp.\ $S\purg F$) and $F\xr\alpha
F'$, then $S\dis F'$ (resp.\ $S\purg F'$).  

\begin{lem}\label{lem:disprops}
  (i)~If $S|F$ is a seed, then $S\dis F$; 
  (ii)~if $S\purg R$, then $S\dis R$.
\end{lem}
\begin{proof}
  \begin{enumerate}[(i)]
  \item By contradiction, if $F\reduct k \alpha_i.S_i$, then
    $F\eqD D[\alpha_i.S_i]$. By law~$(A)$, $S|F\sim S|D[\nil]$ which
    contradicts the minimality hypothesis about $S|F$.
  \item Again, by contradiction, suppose that $R\eqD D[\alpha_i.S_i]$.
    Since, $S\purg R$, there exist $j,D'$ such that
    $S_j\eqD D'[D[\alpha_i.S_i]]$, from which we deduce
    $S~\sim~\prod_{k\neq j}\!\alpha_k.S_k~|~\!a_j.D'[D[\nil]]$ by $(A)$
    (we necessarily have $i\neq j$). This is contradictory with the
    fact that $S$ is a seed. \qed
  \end{enumerate}
\end{proof}

\begin{lem}\label{lem:Rnil}
  $S\sim S|R$ and $S\purg R$ entail $R=\nil$.
\end{lem}
% \begin{lem}\label{lem:Mnil}
%   $S\sim S|\!M$ entails $M=\nil$.
% \end{lem}
%   \comment{$\!\nil$ n'existe pas\dots il faut ecrire $\!\prod_i
%   \alpha_i.M_i$?} 
\begin{proof}
  Suppose by contradiction $R = \alpha.R'|R"$. By
  Lemma~\ref{lem:cancel:finite:part}, we have $S\sim S|\alpha.R'$ and
  $S\sim S|\!\alpha.R'$ by replicating all processes. By
  Lemma~\ref{lem:pouslemma} there exists $i$ such that $S\sim
  \!\alpha.R'|\prod_{k\neq i}\alpha_k.S_k$. Now, since $S\purg R$,
  there exist some $j,D$ such that $S_j\eqD D[\alpha.R']$; if $i=j$,
  we have obtained a smaller seed; otherwise, we use $(A)$ to show
  that $\!\alpha.R'|\!\alpha_j.D[\nil]|\prod_{k\neq i,j}\alpha_k.S_k$
  is a smaller seed.  \qed
\end{proof}

\begin{lem}\label{lem:cancel:finite}
  $\!F_1|F'_1\sim\!F_2|F'_2$ entails $\!F_1\sim\!F_2$.
\end{lem}
\begin{proof}
  Write $S_i$ for the seed of $\!F_i$, $i=1,2$. We have
  $S_1|F'_1\sim S_2|F'_2$. By emptying $F'_1$ on the
  left\footnote{In the present case, `emptying $F'_1$' means playing
    all prefixes in $F'_1$ in the bisimulation game between $S_1|F'_1$
    and $S_2|F'_2$ -- we shall reuse this terminology in some proofs
    below.}, we obtain $S_1\sim S_2|F"_2|R_2$ for some $F"_2,
  R_2$.  Now, by emptying on the right, we get $S_1|R_1\sim
  \!S'_2$.  Injecting the latter equivalence in the one we have
  previously obtained gives
  \begin{mathpar}
    S_1\sim S_1|R_1|F"_2|R_2
    \enspace.
  \end{mathpar}
  If $R_1\not\sim\nil$, we can apply
  Lemma~\ref{lem:cancel:finite:part} to deduce $S_1\sim S_1|R_1$.
  But this gives a contradiction by Lemma~\ref{lem:Rnil}. Hence
  $R_1\sim\nil$, which gives us, since we have established
  $S_1|R_1 \sim S_2$, that $S_1 \sim  S_2$. Finally,
  $\!F_1 \sim\!F_2$.  \qed
\end{proof}

\begin{lem}\label{lem:cancel:FR}
  If $S|F\sim S|R$, $S\dis F$, and $S\purg R$, then $F\sim R$.
\end{lem}
\begin{proof}
  We proceed by induction on the size of $F$. If $F=\nil$, we have
  $R=\nil$ by Lemma~\ref{lem:Rnil}; otherwise, we first prove that $F$
  and $R$ have the same size:
  \begin{itemize}
  \item if $\size F < \size R$, by emptying $F$ on the left-hand side,
    we find $R'\neq\nil$ such that $S|R\reduct{\size F} S|R'$, $S\purg
    R'$ and $S\sim S|R'$; this is contradictory with
    Lemma~\ref{lem:Rnil};
  \item if $\size F > \size R$, by emptying $R$ on the right-hand
    side, we find $R',F'$ with $0<\size F'\leq\size F$ such that
    $S|F\reduct{\size R} S|R'|F'$, $S\purg R'$, $S\dis F'$ and
    $S|R'|F'\sim S$. Then we write $F'=\alpha.F_0|F_1$ and deduce
    $S|\alpha.F_0\sim S$ by Lemma~\ref{lem:cancel:finite:part}; then,
    by firing the $\alpha$ prefix, we find $i$ such that
    $\alpha=\alpha_i$ and $S|F_0\sim S|S_i$. We check that $\size
    F_0<\size F$ so that we can apply the induction hypothesis and
    deduce that $F_0\sim S_i$, whence $\alpha.F_0\sim \alpha_i.S_i$,
    and $\alpha.F_0\eqD \alpha.S_i$ by Thm.~\ref{thm:distrlaw}. This
    is contradictory with $S\dis F$ ($\alpha.F_0$ is a sub-term of
    $F$). 
  \end{itemize}
  This concludes the proof that $F$ and $R$ have the same size. We
  then show that the relation $\{(F,R)\}\cup\,\sim$ is a bisimulation:
  \begin{itemize}
  \item when $F\xr\alpha F'$, we find $R'$ such that $S|R\xr\alpha
    S|R'$ and $S|F'\sim S|R'$; by induction, $F'\sim R'$, and we
    deduce that $R'$ is a derivative of $R$, since otherwise, we would
    have $\size R'\geq \size R=1+\size F'$ which is impossible.
  \item when $R\xr\alpha R'$, either we find $F'$ such that
    $F\xr\alpha F'$ and $S|F'\sim S|R'$, which allows us to close the
    diagram, by induction; or we find $i$ such that $S|S_i|F\sim
    S|R'$. In this case, we empty $R'$ on the right-hand side,
    yielding $R''$ and $F'\neq\nil$ such that $S|R''|F'\sim S$; by
    Lemma~\ref{lem:cancel:finite:part}, $S|F'\sim S$, and $F'\sim\nil$
    by induction, which is contradictory. \qed
  \end{itemize}
\end{proof}

% \noindent\textsl{Nota}: we rely on the finiteness of $F$ and $R$ to prove the
% above result, that is, on the fact that replications are not nested.

\begin{lem}\label{lem:FF}
  If $S|F_1\sim S|F_2$ and $S\dis F_i$ ($i=1,2$), then $F_1\sim F_2$.
\end{lem}
\begin{proof}
  First observe that if $\size{F_1}<\size{F_2}$, then we can empty
  $F_1$ by playing challenges on the left hand side, and we obtain
  $S\sim S|F'_2$ with $F'_2\not\sim\nil$, which is impossible by
  Lemma~\ref{lem:cancel:FR}. Hence $\size{F_1}=\size{F_2}$.
  
  We then show that $\R = \set{(F_1,F_2)~/~ S|F_1\sim S|F_2}$ is a
  bisimulation. If $F_1\xr{\mu}F'_1$, then $S|F_1\xr{\mu}S|F'_1$,
  which by hypothesis entails that $S|F_2$ can answer this challenge.
  By the remark above, $S|F_2$ necessarily answers by firing $F_2$,
  since otherwise we would get equivalent processes with finite parts
  having different sizes. This allows us to show that $F_2$ can
  answer the challenge, and that \R{} is a bisimulation.
  \qed
\end{proof}

\begin{lem}\label{lem:RR}
  $S|R_1\sim S|R_2$ and $S\purg R_i$ ($i=1,2$) entail $R_1\eqD R_2$.
\end{lem}
\begin{proof}
  By Lemma~\ref{lem:FF}, we have $R_1\sim R_2$ ($S\dis R_i$ by
  Lemma~\ref{lem:disprops}(ii)).  We conclude with
  Thm.~\ref{thm:distrlaw}: $R_1$ and $R_2$ are finite processes.  \qed
\end{proof}

\begin{lem}\label{lem:simseed}
  If $S\sim S'$, then $S\eqD S'$.
\end{lem}
\begin{proof}
  Write $S = \prod_{i\leq m} \!\alpha_i.S_i$ and $S' = \prod_{j\leq n}
  \!\alpha'_j.S'_j$, play each prefix on the left-hand side and apply
  Lemma~\ref{lem:RR} to show that there exists a map $\sigma:
  [1..n]\to{}[1..m]$, such that
  $\alpha_i.S_i\eqD\alpha'_{\sigma(i)}.S'_{\sigma(i)}$. This map is
  bijective: otherwise we could construct a smaller seed.
\qed
\end{proof}

\begin{prop}[Uniqueness of seeds]\label{prop:compare:seeds}
  Suppose $P\sim P'$, where $P$ and $P'$ are seeds. Then $P\eqD P'$.
\end{prop}
\begin{proof}
  Write $P \eqD S | F$ and $P' \eqD S' | F'$. As remarked above, $S$
  and $S'$ are necessarily seeds because $P$ and $P'$ are (hence the
  notation). By Lemma~\ref{lem:cancel:finite}, $S\sim S'$, whence
  $S\eqD S'$ by Lemma~\ref{lem:simseed}.
  Necessarily, $S\dis F$ and $S'\dis F'$, which allows us to deduce,
  using Lemma~\ref{lem:FF}, that $F\sim F'$. Finally, $P\eqD P'$, by
  Thm.~\ref{thm:distrlaw}.
  \qed
\end{proof}

%% \comment{inutile vu la preuve}
% Note that in the result above, $P$ and $P'$ may contain some finite
% components at top-level.

% \paragraph{The same proof for $\pi$.}

% In $\pi$, we must take care of $\alpha$-conversion. For instance one
% law becomes $\!a(x).P~|~C[a(x).P] \sim \!a(x).P~|~C[\nil]$ provided none
% of the free names of $a(x).P$ is captured by $C$.

% Similarly, when writing the relation $P\reduct Q$, we can notice that
% $Q$ may have some free names that are not free in $P$. However,
% observe that in Lemma~\ref{lem:Mnil}, we necessarily have
% $\fn{M}\subseteq\fn{S}$.

% We must also think about relation $R\dis S$: morally, $S$ is in
% parallel with $R$, and $R\dis S$ means $\neg(R\reduct a_i(x_i).S_i)$
% where $\fn{a_i(x_i).S_i}\subseteq\fn{R}$.

\section{Rewriting Processes to Normal Forms}
\label{sec:rwt}

% Grammars:
% \begin{mathpar}
%   F ~::=~ \nil \OR F|F\OR \alpha.F
%   \and
%   P ~::=~ F\OR \!\alpha.F\OR P|P
%   \and
%   C\text{ is $P$ with a hole (no $\![\,]$)}
% \end{mathpar}

\begin{defi}[Rewriting, convertibility]\label{def:rewr:conv}
  Any process $P$ induces a relation between processes, written
  \rewr{P}, defined by the following axioms, modulo $\eqD$:
\begin{mathpar}
  C[\alpha.F] \rewr{\!\alpha.F|F'} C[\nil]
  \quad\text{(B1)}
  \and
  \!\alpha.F |\!\alpha.F | P \rewr{Q} ~\!\alpha.F |  P
  \quad\text{(B2)}
%   \and
%   C[\alpha.(P |  (\alpha.P)^k)] \rewr{Q} C[(\alpha.P)^{k+1}]
%   ~\text{(D)}
\end{mathpar}
The reflexive transitive closure of \rewr{P} is written \wrewr{P};
we say that $P$ and $Q$ are \emph{convertible}, written $P\convert Q$,
whenever there exists a process $T$ such that $P\wrewr{T}T$ and
$Q\wrewr{T}T$.
\end{defi}

Example: we can check that process $\!\alpha.(\beta|\alpha.\beta)$ is
normalised into $\!\alpha.\beta$ via the sequence
$\!\alpha.(\beta|\alpha.\beta) \xr{\!\alpha.\beta|\nil}\!\alpha.\beta$
using axiom (B1). This is the way our rewriting relation proceeds to
compute normal forms. In this case, an equational reasoning would be
possible, as follows: $\!\alpha.(\beta|\alpha.\beta) =
\!(\alpha.\beta|\alpha.\beta) = \!\alpha.\beta |\!\alpha.\beta =
\!\alpha.\beta$ (we use the law $(A')$ for the first step).

\begin{lem}\label{lem:rewr:TT}
  If $Q\wrewr{T}T$, then $Q\sim T$.
\end{lem}
\begin{proof}
  By induction over the number of rewrite steps. If this number is
  zero, then this is obvious; suppose now $Q\rewr{T}Q'\wrewr{T}T$. The
  induction hypothesis gives $Q'\sim T$. We reason by cases over the
  axiom that is used to rewrite $Q$ into $Q'$:
  \begin{itemize}
  \item (B1): this means that $Q = C[\alpha.P]$, $Q' = C[\nil]$
    and $T=\!\alpha.P|P'$. From $\!\alpha.P|P'\sim C[\nil]$, we deduce
    $\!\alpha.P|P' \sim C[\alpha.P]$ by Prop.~\ref{prop:simplelaw},
    hence $Q\sim T$.
  \item (B2): we easily have $Q\sim Q'$, hence $Q\sim T$.
    \qed
  \end{itemize}
\end{proof}

\begin{lem}\label{lem:rewr:sn}
  Given $P$, the relation \wrewr{P} terminates.
\end{lem}
\begin{proof}
  The size of processes strictly decreases along reductions.  \qed
\end{proof}

\begin{lem}\label{lem:eqD:rewr}
  For all $P$, either $P\eqD\seed P$, or $P\rewr{\seed P}P'$ for some
  $P'$ s.t.\ $P\sim P'$.
\end{lem}
\begin{proof}
  Write
  \begin{mathpar}
    P = (\prod_i \!\alpha_i.F_i)~|~ F^P\and \text{and}\and \seed{P} =
    (\prod_j \!\alpha_j.S_j)~|~ F^S \enspace,
\end{mathpar}
\noindent and set $S = \prod_j \!\alpha_j.S_j$.
By definition, $P\sim \seed{P}$, which gives, by
Lemma~\ref{lem:cancel:finite}, 
\begin{equation}
\prod_i \!\alpha_i.F_i\sim \prod_j \!\alpha_j.S_j
\enspace.
\label{eq:bang:parts}
\end{equation}
A transition by the left hand side process is answered by the right
hand side process, yielding process
$\prod_i \!\alpha_i.F_i~|~ F_n~\sim~ \!\prod_j \!\alpha_j.S_j~|~ S_m$,
which 
gives, by injecting equivalence~(\ref{eq:bang:parts}), 
$S~|~ F_n\sim S~|~ S_m$.

By Lemma~\ref{lem:FF}, this gives: either $(i)$ $F_n\sim S_m$, which
means by Theorem~\ref{thm:distrlaw} $F_n\eqD S_m$, or $(ii)$
$\neg(S\dis F_n)$ (indeed, $\neg(S\dis S_m)$ is impossible, since this
would allow us to compute a seed having a smaller size than
\seed{P}). In the latter case, $(ii)$, this means that $P$ can be
rewritten using axiom (B1), and the resulting process is bisimilar to
$P$.

Suppose now that we are in case $(i)$ for all possible transitions
from the $\alpha_i.F_i$s, that is, for all $i$, there exists $j$ such that
$\alpha_i.F_i\eqD \alpha_j.S_j$. We observe that the converse (associating a
$i$ to all $j$s) also holds, and that the number of parallel
components in $\prod_i \!\alpha_i.F_i$ is necessarily greater than the
number of components in $S$. In the case where this number is strictly
greater, this means that \rewr{\seed{P}} can be used to rewrite the
left hand side process in~(\ref{eq:bang:parts}), using axiom (B2). In
this case, the resulting process is bisimilar to $P$.

We are left with the case where the two processes have the same number
of components, which entails that they are equated by $\eqD$.

To sum up, we have shown that either $\prod_i \!\alpha_i.F_i$ can be
rewritten, or $\prod_i\!\alpha_i.F_i\eqD S$. In the latter case, we can
inject equivalence~(\ref{eq:bang:parts}) in
$P\sim\seed{P}$, which gives
$S~|~F^P\sim S~|~F^S$. We can apply Lemma~\ref{lem:FF} again,
which gives two possibilities. The first possibility is that
$F^P\sim F^S$, in which case $F^P\eqD F^S$, and finally
$P\eqD \seed{P}$. %, which shows that $P\wrewr{\seed{P}}\seed{P}$.
The second possibility is that $\neg(S\dis F^P)$ (as above,
$\neg(S\dis F^S)$ is not possible since this would allow us to compute
a seed of smaller size). In that case, we can 
rewrite $P$ using (B1), and getting a process bisimilar to $P$.

Finally, either $P\eqD \seed{P}$, or $P$ can be rewritten using
\rewr{\seed{P}}.
\qed
\end{proof}

\begin{prop}\label{prop:seed}
  For all $P$, $P\wrewr{\seed{P}}\seed{P}$.
\end{prop}
\begin{proof}
  Follows by Lemmas~\ref{lem:eqD:rewr} and~\ref{lem:rewr:sn}.
\qed
\end{proof}

\begin{thm}[Characterisation]\label{thm:charac}
  $P\convert Q$ iff $P\sim Q$.
\end{thm}
\begin{proof}
  Suppose $P\convert Q$. By definition, this gives the existence of
  $T$ s.t.\ $P\wrewr{T}T$ and $Q\wrewr{T}T$. We deduce $P\sim Q$ by
  applying Lemma~\ref{lem:rewr:TT} twice and transitivity of $\sim$.
  Hence $\convert\protect{\subseteq}\sim$.
  
  To establish the converse, suppose $P\sim Q$. Write, using
  Proposition~\ref{prop:seed}, $P\wrewr{\seed{P}}\seed{P}$ and
  $Q\wrewr{\seed{Q}}\seed{Q}$. By definition, $P\sim\seed{P}$ and
  $Q\sim\seed{Q}$, which entails $\seed{P}\sim\seed{Q}$.
  This gives by Proposition~\ref{prop:compare:seeds}
  $\seed{P}\eqD\seed{Q}$, which finally gives $P\convert Q$.
  \qed
\end{proof}

This result gives a way to decide whether $P\sim Q$, via
\convert. Indeed, although Definition~\ref{def:rewr:conv} does not
tell how to find process $T$, that allows one to derive $P\convert Q$,
Thm.~\ref{thm:charac} allows us to reduce this problem to checking
whether $Q\rewr{\seed P}\seed P$. For this, it suffices to look for
\seed{P} among all processes of size smaller than \size{P}.

\section{Adding Synchronisations}
\label{sec:comm}

We can now move to a calculus closer to standard CCS, called
\miniccs{}, by instantiating actions with the following grammar, where
$a$ range over a countable set of \emph{names}: actions are either
input or output prefixes.
\begin{mathpar}
  \alpha ~::=~ a\OR \out a
\end{mathpar}
The LTS we obtain with this definition is not that of CCS: we need to
add the following rules for synchronisations, where $\tau$ is the
label for internal moves.
\begin{mathpar}
  \inferrule{P\xr a P' \and Q\xr{\out a} Q'}{P|Q \xr\tau P'|Q'}
\and
  \inferrule{P\xr{\out a} P' \and Q\xr{ a} Q'}{P|Q \xr\tau P'|Q'}
\end{mathpar}

In doing so, we change the notion of strong bisimilarity: the standard
CCS bisimilarity, that we shall denote using $\sims$, tests internal
moves while our notion of bisimilarity ($\sim$) plays visible
challenges only. Therefore, we have $\sims \,\subseteq\, \sim$.

The following result says that \convert{} is actually enough to
capture strong bisimilarity on \miniccs. As a consequence, we do not
need to test $\tau$ transitions to obtain the discriminating power of
$\sims$.

\begin{prop}\label{prop:equivccs}
  Let $P$ and $Q$ be two processes. Then $P\sims Q$ if and only if
  $P\convert Q$.
\end{prop}
\begin{proof}
  By Thm.~\ref{thm:charac} and the above remark, it suffices to show
  that $\convert\,\subseteq\,\sims$. This amounts to check that the
  distribution law and Prop.~\ref{prop:simplelaw} are valid for
  $\sims$: we just need to check that silent challenges can be
  answered in the corresponding bisimulation candidates.  \qed
\end{proof}

Note that the $\tau$ prefix is not included in this presentation of
\ccs; indeed, adding $\tau$ to the syntax of actions $(\alpha)$ would
\emph{a priori} break the inclusion $\sims \,\subseteq\, \sim$: tests
performed by $\sim$ on $\tau$-transitions would be too restrictive,
since the only way to answer would be to use a $\tau$ prefix --
synchronisations would not be allowed. In the light of
Prop.~\ref{prop:equivccs}, we actually believe that $\tau$ prefixes
could be added, that is, that they are played in one-to-one
correspondence in bisimilarity games.
%  ; yet, we did not manage to prove
% it.

\medskip

We conclude this section by proving that bisimilarity is closed under
substitutions in \miniccs.  We use $\sigma$ to range over
\emph{substitutions}, that are functions mapping names to names; we
write $P\sigma$ for the process we obtain by applying $\sigma$ on all
names of $P$.

\begin{prop}[$\sim$ is closed under substitutions in \miniccs]
  \label{prop:subst:clos:ccs}
  If $P\sim Q$, then for any $\sigma$, $P\sigma\sim
  Q\sigma$.
\end{prop}
\begin{proof}
  We show the property for \convert. Suppose $P\convert Q$, which
  gives the existence of $T$ such that, in particular, $P\wrewr{T}T$.
  By inspecting the shape of axioms (B1) and (B2), and reasoning by
  induction over the number of rewrite steps, we can deduce that
  $P\sigma\wrewr{T\sigma}T\sigma$. Similarly,
  $Q\sigma\wrewr{T\sigma}T\sigma$. Hence
  $P\sigma\protect{\convert}Q\sigma$.
\qed
\end{proof}

\section{Concluding Remarks}\label{sec:piccl}

\subsection{Extending our Characterisation}
\label{sec:ccl}

In absence of restriction in the calculus, it is easy to see that
applying replication to prefixed processes only is of no harm in terms
of expressiveness, because of the following rather standard structural
congruence laws for $\equiv$ (which are of course valid strong
bisimilarity laws):
\begin{mathpar}
  \!(P|Q) \equiv  \!P |\!Q
  \and
  \!!P \equiv \!P
  \and
  \!\nil \equiv \nil
\end{mathpar}

We have started investigating the question of characterising $\sim$ in
the case where replication is not at top-level (but where nested
replications -- that is, replications that occur under replications --
are forbidden). The law
\begin{mathpar}
  \alpha.C[\!\alpha.C[\nil]] ~=~ \!\alpha.C[\nil]
\end{mathpar}
\noindent seems important to capture $\sim$ in this setting. We do not
know at the moment whether it is sufficient to characterise $\sim$.

Handling nested replications seems even more challenging.

\subsection{Congruence of Strong Bisimilarity in the $\pi$-calculus}
\label{sec:pi}

Because of the input prefix, congruence of strong bisimilarity
requires closure of this relation under substitutions. 
In presence of sum, this property fails; as~\cite{SW01} shows, this is
also the case as soon as replication and restriction are present in
the calculus (in absence of sum).

\cite{hirschkoff:pous:lmcs:08} shows that congruence holds when we
renounce to replication, that is, in the sub-calculus that features
input and output prefixes, parallel composition and restriction.

Our investigations have convinced us that the same holds if instead we
renounce to restriction: we believe that the reasoning seen above can
be ported to the following subset of the $\pi$-calculus:
\begin{mathpar}
  F ~::=~ \nil \OR F|F\OR a(x).F\OR \outm{a}{b}.F
  \and
  P ~::=~ F\OR \!a(x).F\OR P|P
\end{mathpar}
%  (some care has
% to be taken with $\alpha$-equivalence). 
The analogue of Prop.~\ref{prop:subst:clos:ccs} gives us closure under
substitutions of strong bisimilarity, which in turn yields congruence.
Note that when bisimilarity is closed under substitutions, the ground,
early and late versions of the equivalence coincide. To adapt our
method from \miniccs{} to the $\pi$-calculus, we work with ground
bisimilarity.

\bibliography{refs.bib}
\bibliographystyle{plain}

\end{document}